\documentclass[letterpaper, 10 pt, conference]{ieeeconf}  
\usepackage[bottom=0.685in,top=0.681in,left=0.625in,right=0.625in]{geometry}

\IEEEoverridecommandlockouts                              
\usepackage{amsmath}
\usepackage{amssymb}
\usepackage{mathtools}

\usepackage{amsthm}

\usepackage{physics}
\usepackage{braket}
\usepackage{dsfont}

\usepackage{subfigure}
\usepackage{caption}

\usepackage{epstopdf}
\usepackage{multirow}
\usepackage{float}
\usepackage{mathtools}
\usepackage{graphicx}
\usepackage{array}
\usepackage{mdwmath}
\usepackage{mdwtab}
\usepackage{eqparbox}
\usepackage{nomencl}
\usepackage{cite}
\usepackage{algorithm}
\usepackage{algorithmic}
\usepackage{makeidx}
\usepackage{ifthen}
\usepackage{subfigure}
\usepackage{caption}

\usepackage{algorithm,algorithmic}

\theoremstyle{plain}
\newtheorem{theorem}{Theorem}[section]
\newtheorem{lemma}[theorem]{Lemma}

\newtheorem{proposition}[theorem]{Proposition}

\newtheorem{assumption}[theorem]{Assumption}
\theoremstyle{definition}
\newtheorem{definition}[theorem]{Definition}

\newtheorem{remark}[theorem]{Remark}

 \title{\LARGE \bf Projected Variational Quantum Extragradient for Zero-Sum Games \vspace{-0.4cm}}
\author{Duong The Do$^{1}$, Matthew Aldridge$^{2}$ and Duong Tung Nguyen$^{1}$ \vspace{-0.55cm}
}

\begin{document}

\maketitle
\thispagestyle{empty}
\pagestyle{empty}

\begin{abstract}
 We propose a projected variational quantum extragradient (VQEG) framework for computing approximate Nash equilibria in two-player zero-sum matrix games. Mixed strategies are parameterized as Born distributions of parameterized quantum circuits (PQCs), transforming the classical bilinear saddle-point problem into a smooth but generally nonconvex–nonconcave min–max optimization in circuit-parameter space. The expected payoff is expressed as the expectation of a diagonal observable, enabling gradient evaluation via the parameter-shift rule and compatibility with shot-based quantum hardware.
 To support arbitrary game sizes, we introduce a dominated embedding that maps $(m,n)$ games to qubit-compatible power-of-two dimensions while preserving equilibrium structure. We then develop a projected extragradient method using stochastic gradient estimates derived from finite measurement shots, and establish variance bounds scaling as $O(1/S)$ with respect to the number of measurement shots $S$, along with convergence to approximate first-order stationarity under standard assumptions. Since stationarity does not guarantee equilibrium optimality, we evaluate performance using the game-space Nash gap. Numerical results demonstrate high-precision solutions on structured instances up to $32 \times 32$, while highlighting challenges in unstructured settings.
\end{abstract}

\section{INTRODUCTION}
\label{sec:Introduction}
Two-player zero-sum games are a widely used model for competitive decision-making and arise in a wide range of applications, including cybersecurity \cite{YanZhou_2021}, adversarial machine learning \cite{Moghadam_2022}, robust control \cite{ZongyangJiang2025}, online learning and operations research \cite{Elif2025}, and reinforcement learning for competitive multi-agent systems \cite{YuanhengZhu_2020}. The central solution concept in these games is the Nash equilibrium (NE), which characterizes a strategy profile in which no player can improve its payoff through unilateral deviation \cite{Binmore2007}. In finite zero-sum games, NE always exist and correspond to minimax and maximin strategies \cite{MartinOsborne2004}. This equivalence connects equilibrium computation with saddle-point optimization and motivates iterative methods for solving zero-sum games.

Although equilibria in zero-sum games can be computed exactly via linear programming, many modern applications involve large-scale or oracle-based settings where direct methods become impractical. This has motivated the development of scalable first-order methods, including extragradient, mirror-prox, and optimistic gradient algorithms, which exploit the saddle-point structure and provide strong convergence guarantees \cite{Nemirovski_2004, DaskalakisPanageas2019, Mokhtari2020, Cen2024, Feng2024}. 
In parallel, no-regret learning dynamics offer an alternative approach by interpreting equilibrium computation as repeated play \cite{Abe2023, Cai2025}. 
There is also growing interest in last-iterate convergence to obtain stronger practical guarantees for equilibrium computation, since practical performance depends on the actual strategies generated by the algorithm rather than only on averaged iterates \cite{Abe2023, Feng2024, Chen2024}. Also, residual-based methods that directly optimize equilibrium measures, such as the duality gap, have also been studied for approximating equilibria in zero-sum games \cite{Fasoulakis2025}. These highlight the importance of designing scalable algorithms that can exploit problem structure while remaining practical for complex game settings. 

Despite these advances, most existing methods operate directly in the strategy space, which can become high-dimensional and computationally expensive.
These challenges have motivated the study of alternative computational frameworks that can represent and optimize player strategies.
In this work, we propose a variational quantum approach for finite two-player zero-sum games. Our approach leverages PQCs to represent mixed strategies via Born distributions, inducing a smooth mapping from circuit parameters to probability simplices \cite{Wierichs2022}. Under this representation, the classical bilinear objective can be expressed as the expectation of a diagonal observable, leading to a parametric min–max optimization problem in circuit-parameter space. However, this transformation introduces nonlinearity, resulting in a generally nonconvex-nonconcave saddle-point problem.


To address this challenge, we develop a projected variational quantum extragradient (VQEG) method that performs stable saddle-point updates in parameter space using parameter-shift gradient estimators compatible with shot-based quantum hardware. To enable qubit-compatible representations, we introduce a dominated embedding that maps arbitrary $(m,n)$ games to power-of-two dimensions while preserving equilibrium structure. We further analyze the stochastic optimization process and establish convergence to approximate first-order stationarity under standard smoothness and bounded-variance assumptions.
Since stationarity in parameter space does not necessarily imply equilibrium optimality in the original game, we evaluate the induced strategies using the game-space Nash gap as an empirical certificate.
Our main contributions are summarized below.
\begin{itemize}
    \item We introduce a variational quantum parameterization of mixed strategies for zero-sum matrix games using PQC-induced Born distributions and the corresponding expected payoff is expressed in an observable form.
    \item We develop an equilibrium-preserving dominated embedding that enables qubit-compatible representations of arbitrary game sizes. 
    \item We propose a \emph{variational quantum extragradient} (VQEG) algorithm that updates the circuit parameters through a two-step extragradient procedure using parameter-shift gradient estimators, and we analyze the resulting circuit-evaluation and shot complexity.
    \item We prove convergence to approximate first-order \emph{stationarity} in parameter space under standard smoothness and bounded-variance assumptions, and we assess equilibrium in the original game space via the Nash gap.

\item Experimental results are conducted to validate the efficacy of the proposed VQEG framework. 
\end{itemize}

The rest of the paper is organized as follows. Section~\ref{sec:Prelimiary} reviews preliminaries on finite two-player zero-sum games and the equilibrium certificates used in this work. Section~\ref{sec:QVEG} presents the proposed VQEG method. 
Section~\ref{sec:Convergence} provides the convergence analysis. 
Section~\ref{sec:Results} presents numerical results, and Section~\ref{sec:Conclusion} concludes the paper.

\section{PRELIMINARIES: ZERO-SUM GAMES AND EQUILIBRIUM CERTIFICATES}
\label{sec:Prelimiary}
\subsection{Finite Zero-sum Games}

We consider a finite two-player zero-sum matrix game defined by a payoff matrix $A \in \mathbb R^{m\times n}$ , where $m$ and 
$n$ denote the numbers of pure strategies (actions) available to the row player (Player 1) and the column player (Player 2), respectively. The column player’s payoff matrix is given by 
$-A$.
  Let $ \Delta_s := \{p \in \mathbb R^s_{\ge 0}: \mathbf{1}^\top p = 1\}, s \in \left\{m, n \right\}$ denote the probability simplex. The players choose mixed strategies $x \in \Delta_m$ and $y \in \Delta_n$, and the expected payoff to  the row player is given by $f(x,y) \;:=\; \left\langle x, Ay \right\rangle = x^\top A y$.
It is easy to see that the function $f(x,y)$ is bilinear and bounded between the minimum and maximum entries of $A$, i.e., we have: 
    \begin{align}
    \label{le:bilinear_payoff}
        & \min_{i \in [m],\, j \in [n]} A_{ij} \;\le\; f(x,y) \;\le\; \max_{i \in [m],\, j \in [n]} A_{ij}.
    \end{align}
    
\begin{definition}[Nash equilibrium]
A strategy profile $(x^\star,y^\star)$ is an NE (equivalently, a saddle point) if:
    \begin{align}
    \label{eq:NE_def}
        & f(x^\star, y) \ge f(x^\star, y^\star) \ge f(x,y^\star), \; \forall x \in \Delta_m, \forall y\in\Delta_n.
    \end{align}
\end{definition}
By von Neumann's minimax theorem, the game admits a value, given by:
\begin{align}
\label{eq:game_value}
    &  \max_{x\in\Delta_m}\min_{y\in\Delta_n} f(x,y) = \min_{y\in\Delta_n}\max_{x\in\Delta_m} f(x,y) = v^\star, 
\end{align}
and every Nash equilibrium attains this value. Thus, computing equilibria in zero-sum games is equivalent to solving a saddle-point problem $\max_{x\in\Delta_m}\min_{y\in\Delta_n} f(x,y)$.

\subsection{Approximate Equilibrium and Nash Gap}
For any strategy profile $(x,y) \in \Delta_m \times \Delta_n$, let $\alpha(y)=\max_{x'\in\Delta_m} \langle x', Ay \rangle$ denote the row player's best-response payoff against $y$ and $\beta(x) = \min_{y' \in \Delta_n} \langle x, Ay' \rangle$ be the column player's best-response payoff against $x$.
The unilateral deviation gains for the row and column players are given by:
\begin{align}
\label{eq:dev_gain}
    \delta_{\mathrm{r}}(x,y) \!=\! \alpha(y) \!-\! \left\langle x, Ay \right\rangle, \; \delta_{\mathrm{c}}(x,y) \!=\! \left\langle x, Ay \right\rangle \!-\! \beta(x),
\end{align}
where $\delta_{\mathrm{r}}(x,y)$ is the maximum payoff improvement available to the row player by deviating from $x$ while the column player keeps $y$ fixed. Similarly, $\delta_{\mathrm{c}}(x,y)$ is the maximum improvement available to the column player by deviating from $y$ while the row player keeps $x$ fixed. 
By weak minimax duality, for all $(x,y)$ we have $\beta(x) \leq v^\star \leq \alpha(y)$. 
\begin{definition}[Nash gap]
    For any $(x,y) \in \Delta_m \times \Delta_n$, the Nash gap (or duality gap) is defined as:
    \begin{align}
    \label{eq:NE_gap}
        & \mathcal{G}(x,y) := \alpha(y) - \beta(x), 
    \end{align}
which can be equivalently written as
    \begin{align}
    \label{eq:NE_gap_eq}
        & \mathcal{G}(x,y) = \max_{i \in [m]} (Ay)_i - \min_{j\in[n]} (x^\top A)_j.
    \end{align}
\end{definition}

\begin{definition}[$\varepsilon$-NE]
    $(x^\star,y^\star)$ is an $\varepsilon$-NE if no player can gain more than $\varepsilon$ by unilateral deviation, i.e., we have:
    \begin{align}
        \begin{array}{cc}
        \left\langle x^\star, Ay^\star \right\rangle \geq & \max_{x \in \Delta_m} \left\langle x, Ay^\star \right\rangle - \varepsilon, \\
        \left\langle x^\star, Ay^\star \right\rangle \leq & \min_{y \in \Delta_n} \left\langle x^\star, Ay \right\rangle + \varepsilon,
        \end{array}
    \end{align}
    In other words, $\max\left\{\delta_{\mathrm{r}}(x^\star,y^\star), \delta_{\mathrm{c}}(x^\star,y^\star) \right\} \le \varepsilon$.
\end{definition}

It is observed that the deviation of both players can be expressed by using the Nash gap as, $\mathcal{G}(x,y) = \delta_{\mathrm{r}}(x,y) + \delta_{\mathrm{c}}(x,y)$, which satisfies $\max \{\delta_{\mathrm{r}}, \delta_{\mathrm{c}}\} \le \mathcal{G}(x,y) \le 2\max \{\delta_{\mathrm{r}}, \delta_{\mathrm{c}}\}$. Thus, the gap  upper-bounds both players’ incentives to deviate serving as a computable certificate of approximate optimality for zero-sum games \cite{Nemirovski_2004}. In particular, $\mathcal{G}(x,y) = 0$ if and only if $(x,y)$ is an exact NE, and $\mathcal{G}(x,y) \le \varepsilon$ implies $(x,y)$ is an $\varepsilon$-NE. 


\section{VARIATIONAL QUANTUM EXTRAGRADIENT}
\label{sec:QVEG}
\subsection{Equilibrium-Preserving Dominated Embedding}
Variational quantum circuits typically represent distributions over $2^q$ computational basis outcomes. To support arbitrary game sizes $(m, n)$, we propose to embed the original game into a $(M, N)$ game, where $M := 2^{\lceil\log_2 m\rceil}$ and $N := 2^{\lceil\log_2 n\rceil}$. Rather than using the naive zero-padding that may change the equilibrium value, we use an embedding that preserves the original game by ensuring the extra (\textit{dummy}) actions such that all added actions are \emph{strictly dominated} and therefore receive zero probability at equilibrium. 
Let $\|A\|_\infty = \max_{i,j}|A_{ij}|$, and choose 
$C>\|A\|_\infty$. We define the dominated embedding $\widetilde A\in\mathbb{R}^{M\times N}$ as follows:
\begin{align}
\label{eq:dom_embed}
    \widetilde A_{ij} = 
    \begin{cases} 
        A_{ij}, & i\le m,\ j\le n,\\
        -C, & i>m,\ j\le n,\\
        +C, & i\le m,\ j>n,\\
        0, & i>m,\ j>n.
    \end{cases}
\end{align}

\begin{proposition}[Dominated embedding equivalence]
\label{prop:matrix_equiv}
    Let $C>\|A\|_\infty$ and let $\widetilde A$ be defined by~\eqref{eq:dom_embed}. Then the following statements hold:
    \begin{enumerate}
        \item Every NE $(\widetilde x^\star, \widetilde y^\star) \in \Delta_M \times \Delta_N$ of the embedded game satisfies $\sum_{i>m}\widetilde x^\star_i = 0, \, \sum_{j>n}\widetilde y^\star_j = 0$.
        \item The restriction $(x^\star, y^\star) = (\widetilde x^\star_{1:m}, \widetilde y^\star_{1:n})$ is a NE of the original game $(A,\Delta_m, \Delta_n)$.
        \item Conversely, if $(x^\star,y^\star)$ is a NE of $(A, \Delta_m, \Delta_n)$, then the extension $(\widetilde x^\star, \widetilde y^\star)$ created by appending zeros, $\widetilde x^\star = (x^\star, 0) \in \Delta_M$ and $\widetilde y^\star = (y^\star, 0) \in \Delta_N$, is a NE of $(\widetilde A, \Delta_M, \Delta_N)$.
    \end{enumerate}
\end{proposition}
\begin{proof}
    Please refer to our technical report \cite{duong2026vqegtech}. 
\end{proof}

\subsection{PQC Parameterization of Mixed Strategies}
Let $\mathcal{H}_r \simeq \mathbb{C}^{M}$ and $\mathcal{H}_c \simeq \mathbb{C}^{N}$ denote the Hilbert spaces associated with the row and column players, respectively. We represent each player's mixed strategy as a Born distribution generated by a PQC. Specifically, the row player uses an unitary $U_{\theta}:\mathcal{H}_r \to \mathcal{H}_r$ acting on $q_r=\lceil \log_2 M \rceil$ qubits with parameters $\theta\in\mathbb{R}^{d_r}$, while the column player uses a unitary circuit $V_{\phi}:\mathcal{H}_c \to \mathcal{H}_c$ acting on $q_c=\lceil \log_2 N \rceil$ qubits with parameters $\phi \in \mathbb{R}^{d_c}$.

For circuit construction, in our implementation, both players employ a hardware-efficient ansatz composed of alternating layers of single-qubit rotations and entangling gates. Specifically, we apply $\mathrm{R}_y$ and $\mathrm{R}_z$ rotations to every qubit, thereby enabling full local single-qubit expressivity, and subsequently implement a $\mathrm{CZ}$ gate in a ring topology to establish nearest-neighbor entanglement across the register.
The single-qubit rotation $\mathrm{R}_y$ and $\mathrm{R}_z$ are defined as: 
\begin{align}
\label{eq:Ry}
    & \!\! R_y(\varphi) = \!
    \begin{bmatrix} 
        \cos(\frac{\varphi}{2}) & \!\!\!-\sin(\frac{\varphi}{2}) \\ 
        \sin(\frac{\varphi}{2}) & \!\!\!\cos(\frac{\varphi}{2}) 
    \end{bmatrix},
    R_z(\theta) = \!
    \begin{bmatrix} 
        e^{-i\frac{\theta}{2}} & \!\!0 \\ 
        0 & \!\!e^{i\frac{\theta}{2}} 
    \end{bmatrix} \!.
\end{align}
For a circuit with $q$ qubits and $L$ layers, this ansatz contains $2qL$ trainable parameters per player. Thus, the parameter count grows linearly with both the number of qubits and the circuit depth, providing a controlled trade-off between expressivity and implementation on near-term devices.
Additionally, for parameterized unitaries $U_\theta \!\in\! \mathcal{U}(\mathcal{H}_r)$ and $V_\phi \! \in \! \mathcal{U}(\mathcal{H}_c)$, we define the corresponding density operators:
\begin{align}
\label{eq:qmp_states}
    & \rho_\theta := U_\theta \rho_0 U_\theta^\dagger \in \mathbb{C}^{M \times M}, \; \sigma_\phi := V_\phi \sigma_0 V_\phi^\dagger \in \mathbb{C}^{N \times N},
\end{align}
where the initial states are chosen as pure states, $\rho_0 = \sigma_0 =|0\rangle\langle 0|$. Thus, each player’s strategy is represented by a variationally prepared quantum state.
Let $\{\Pi_i\}_{i=1}^{M}$ and $\{\Omega_j\}_{j=1}^{N}$ be computational-basis projectors, defined by $\Pi_i := |i \rangle \langle i| \in \mathbb{C}^{M \times M}$ and $\Omega_j := |j \rangle \langle j| \in \mathbb{C}^{N \times N}$. 
Measurement in the computational basis induces classical mixed strategies $x_\theta \in \Delta_M$ and $y_\phi \in \Delta_N$, whose components are given by the corresponding outcome probabilities:
\begin{align}
\label{eq:born_strat}
    & x_\theta(i) = \Tr(\rho_\theta \Pi_i), \ y_\phi(j) = \Tr(\sigma_\phi \Omega_j).
\end{align}
Hence, $x_\theta \in \Delta_M$ and $y_\phi \in \Delta_N$ are probability distributions over the pure action space, parameterized smoothly by $\theta$ and $\phi$ through the variational circuits. The quantum circuits thus provide smooth parameterizations of the classical simplices. 

For any $y\in\Delta_N$ and $x\in\Delta_M$, we further define the diagonal operator on the Hilbert spaces $\mathcal{H}_r$ and $\mathcal{H}_c$ as:
\begin{align}
\label{eq:qmp_Dy_Ex}
    & D(y) = \! \sum_{i \in [M]} \! (\widetilde A y)_i \Pi_i, \; \; E(x) = \! \sum_{j \in [N]} \! (\widetilde A^\top x)_j \Omega_j.
\end{align}
These observables encode the expected payoff contributions induced by the opponent’s mixed strategy. The embedded expected payoff induced by $(\theta,\phi)$ can then be written as:
\begin{align}
\label{eq:qmp_obj}
    & \! \mathcal{L}(\theta,\phi) \!=\! \langle x_{\theta}, \widetilde{A}\,y_{\phi} \rangle \!= \Tr\!\big( \rho_\theta D(y_\phi) \big) \!= \Tr\!\big( \sigma_\phi E(x_\theta) \big).
\end{align}
Thus, the bilinear classical payoff is embedded as the expectation value of a circuit-parameterized diagonal observable circuit-prepared quantum state, offering a correspondence between the game matrix and variational quantum optimization. 

Furthermore, the row player aims to maximize $\mathcal{L}$ with respect to $\theta$, while the column player seeks to minimize it with respect to $\phi$. This leads to the following parametric saddle-point problem:
\begin{align}
\label{eq:minmax_param}
    & \max_{\theta\in\Theta}\ \min_{\phi\in\Phi}\ \mathcal{L}(\theta,\phi),
\end{align}
where $\Theta\subset\mathbb{R}^{d_r}$ and $\Phi\subset\mathbb{R}^{d_c}$ are bounded sets used to ensure stability and enable projected gradient updates.

Let's denote $\omega = (\theta, \phi) \in \mathbb{R}^d$ with $d=d_r+d_c$ and define the parametric saddle operator as:
\begin{align}
\label{eq:param_operator}
    G(\omega) 
    &= \begin{bmatrix}
            \nabla_\theta \mathcal{L}(\theta,\phi) \\[4pt]
            -\nabla_\phi \mathcal{L}(\theta,\phi)
    \end{bmatrix}
    = \begin{bmatrix}
        \left(\frac{\partial \mathcal{L}}{\partial \theta_1}, \dots, \frac{\partial \mathcal{L}}{\partial \theta_{d_r}}\right)^\top \\[6pt]
        -\left(\frac{\partial \mathcal{L}}{\partial \phi_1}, \dots, \frac{\partial \mathcal{L}}{\partial \phi_{d_c}}\right)^\top
    \end{bmatrix}.
\end{align}
A stationary point $\omega^\star$ satisfies $G(\omega^\star) = 0$, corresponding to a first-order saddle point of the variational game. Since the mappings $\theta \mapsto x_\theta$ and $\phi \mapsto y_\phi$ induced by the variational circuits are typically nonlinear, the objective function $\mathcal{L}(\theta,\phi)$ is generally nonconcave in $\theta$ and nonconvex in $\phi$. Consequently, the associated gradient $G$ is generally non-monotone. Accordingly, VQEG can be viewed as a stabilizing first-order method for solving the parametric min--max problem, which is associated with extragradient and optimistic gradient methods.

\subsection{Parameter-shift gradient estimation}


Let's define $\theta_k^{\pm} = \theta \pm \frac{\pi}{2}e_k$ as shifted row parameters, where $e_k$ denotes the $k^{\text{th}}$ standard basis vector in the row-player parameter space. For each coordinate $\theta_k$, the parameter-shift rule can be then defined as follows:
\begin{align}
\label{eq:ps_theta}
    \frac{\partial \mathcal{L}(\theta,\phi)}{\partial \theta_k}
    &= \frac12 \Big( \mathcal{L}(\theta_k^{+}, \phi) - \mathcal{L}(\theta_k^{-}, \phi) \Big) \\
    &= \frac12 \Big( \Tr(\rho_{\theta_k^{+}} D(y_\phi))- \Tr(\rho_{\theta_k^{-}} D(y_\phi)) \Big), \nonumber 
\end{align}
where $\rho_{\theta_k^{\pm}}$ are the shifted operators obtained from the corresponding parameterized density \eqref{eq:qmp_states}. Equivalently:
\begin{align}
\label{eq:ps_theta_final}
    \frac{\partial \mathcal{L}(\theta,\phi)}{\partial \theta_k} & = \frac12 \Big( \Tr(U_{\theta_k^{+}} \ket{0}\! \bra{0}U_{\theta_k^{+}}^\dagger D(y_\phi)) 
    \\
    & \quad - \Tr(U_{\theta_k^{-}}\ket{0}\!\bra{0}U_{\theta_k^{-}}^\dagger D(y_\phi)) \Big). \nonumber 
\end{align}
Similarly, for each column-player parameter $\phi_\ell$, we can obtain the partial derivative as follows:
\begin{align}
\label{eq:ps_phi_final}
    \frac{\partial \mathcal{L}(\theta,\phi)}{\partial \phi_\ell} 
    & = \frac12 \Big( \Tr(V_{\phi_\ell^{+}} \ket{0} \! \bra{0} V_{\phi_\ell^{+}}^\dagger E(x_\theta)) 
    \\
    & \quad - \Tr(V_{\phi_\ell^{-}} \ket{0} \! \bra{0} V_{\phi_\ell^{-}}^\dagger E(x_\theta) ) \Big), \nonumber 
\end{align}
where $\phi_\ell^{\pm}=\phi\pm\frac{\pi}{2}e_\ell$, and $e_\ell$ denotes the $\ell^{\text{th}}$ standard basis vector in the column-player parameter space.
Each evaluation of $\mathcal{L}$ in~\eqref{eq:ps_theta_final}--\eqref{eq:ps_phi_final} corresponds to the expectation of a diagonal observable and can be estimated using shot-based sampling of the corresponding circuit. The estimator of $\mathcal{L}$ is unbiased, and therefore the resulting gradient estimator obtained via the parameter-shift rule is also unbiased.


\emph{Shot-based estimator:}
Suppose that each expectation value is estimated using $S$ measurement shots. For a fixed shifted parameter vector $\theta_k^{\pm}$, let $\widehat{p}i^{\pm} = \frac{n_i^{\pm}}{S}$ denote the empirical frequency of outcome $i \in [M]$ obtained by measuring the row circuit $U{\theta_k^{\pm}}|0 \rangle$ in the computational basis, where $n_i^{\pm}$ is the number of occurrences of outcome $i$ ($n_i^{\pm} \sim \mathrm{Binomial}(S, x_{\theta_k^{\pm}}(i))$ ). It follows that:
\begin{align}
\label{eq:freq_x_theta}
    & \mathbb{E}[\widehat p_i^{\pm}] = x_{\theta_k^\pm}(i),\\
    & \mathrm{Var}(\widehat p_i^{\pm}) = \frac{1}{S} x_{\theta_k^\pm}(i) \big(1-x_{\theta_k^\pm}(i) \big) \le \frac{1}{4S},
\end{align}
implying the frequency is an unbiased estimator. An unbiased estimator of the shifted payoff $\mathcal{L}(\theta_k^{\pm}, \phi)$ is, thus:
\begin{align}
\label{eq:qmp_obj_estimate}
    \widehat{\mathcal{L}}(\theta_k^{\pm}, \phi) = \sum_{i=1}^{M}\widehat p_i^{\pm}\,(\widetilde A y_\phi)_i = \sum_{j=1}^{N}y_\phi(j)\sum_{i=1}^{M}\widehat p_i^{\pm}\widetilde A_{ij},
\end{align}
Since $\mathbb{E}[\widehat{p}_i^{\pm}] = x_{\theta_k^{\pm}}(i)$, it follows by linearity that $\mathbb{E}\!\left[\widehat{\mathcal{L}}(\theta_k^\pm,\phi)\mid y_\phi \right] = \sum_{i=1}^{M} \mathbb{E} [\widehat p_i^\pm](\widetilde A y_\phi)_i = \mathcal{L}(\theta_k^{\pm}, \phi)$. Moreover, $\widehat{\mathcal{L}}(\theta_k^\pm, \phi) = \sum_{i=1}^M \widehat p_i^\pm(\widetilde A y_\phi)_i$
is a sample average of bounded weights $(\widetilde A y_\phi)_i$ under $S$ shots.
Since $|(\widetilde A y_\phi)_i| \le \|\widetilde A y_\phi \|_\infty \le \|\widetilde A \|_\infty$ for all $i$, it follows that:
\begin{align}
\label{eq:qmp_obj_estimate_var}
    & \mathrm{Var} \! \left(\widehat{\mathcal{L}}(\theta_k^\pm, \phi) \right) 
    \le \frac{\|\widetilde A y_\phi\|_\infty^2}{S}
    \le \frac{\|\widetilde A\|_\infty^2}{S}.
\end{align}
Additionally, the corresponding empirical parameter-shift gradient estimator is given by: 
\begin{align}
\label{eq:ps_theta_estimate}
    & \widehat{\partial_{\theta_k}\mathcal{L}}(\theta, \phi) = \frac{1}{2} \Big( \widehat{\mathcal{L}}(\theta_k^{+}, \phi) - \widehat{\mathcal{L}}(\theta_k^{-}, \phi) \Big).
\end{align}
Since each shifted payoff estimator $\widehat{\mathcal{L}}(\theta_k^\pm, \phi)$ is unbiased, linearity of expectation implies that the parameter-shift gradient estimator is also unbiased. In particular, we have:
\begin{align}
\label{eq:var_grad_row}
    & \mathbb{E} \! \left[\widehat{\partial_{\theta_k} \mathcal{L}}(\theta, \phi) \mid y_\phi \right] = \partial_{\theta_k} \mathcal{L}(\theta,\phi).
\end{align}
Under independent shot sampling with $S$ shots per shifted circuit evaluation, the conditional variance satisfies: 
\begin{align}
\label{eq:var_grad_row}
    \mathrm{Var} \! \left(\widehat{\partial_{\theta_k}\mathcal{L}}(\theta, \phi) \right)
    &= \mathrm{Var} \! \left(\tfrac12(\widehat{\mathcal{L}}(\theta_k^+, \phi) - \widehat{\mathcal{L}}(\theta_k^-, \phi)) \right) \\
    &\le \frac{\|\widetilde A\|_\infty^2}{2S}. \nonumber
\end{align}

Similarly, 
we can obtain the variance of the parameter-shift estimator for the column player as follows:
\begin{align}
    & \mathrm{Var} \! \left(\widehat{\partial_{\phi_\ell} \mathcal{L}}(\theta, \phi) \mid x_\theta \right) \le \frac{\|\widetilde A\|_\infty^2}{2S}.
\label{eq:unbiased_var_grad_col}
\end{align}
With $S$ shots per evaluation and bounded payoff weights (e.g., $\|\widetilde A\|_\infty<\infty$), the variance of each payoff estimate scales as $O(1/S)$, and the variance of each gradient coordinate scales as $O(1/S)$.
Consequently, for the saddle-gradient estimator $\widehat G(\omega)$ with $d$ parameters, the following holds:
\begin{align}
    \mathbb{E}\!\left[\|\widehat G(z)-G(z)\|_2^2 \mid z\right] \le \frac{d}{2S}\,\|\widetilde A\|_\infty^2.
\end{align}

\subsection{Projected Quantum Extractgradient Update}
Extragradient methods are designed for saddle-point and adversarial optimization problems. Rather than updating parameters using only the gradient at the current iterate, they employ a two-step predictor–corrector scheme. Specifically, a predictor step is first computed using the current gradient, and the final update is then obtained using the gradient evaluated at this extrapolated point. This mechanism mitigates cycling behavior and improves stability compared to standard single-step gradient methods.

Let $\mathcal{W}=\Theta\times\Phi \subset \mathbb{R}^d$ be a convex compact set, and let $\Pi_{\mathcal W}$ denote the Euclidean projection onto $\mathcal W$. Given the current iterate $\omega_t=(\theta_t,\phi_t)$, a stepsize $\eta>0$, and a parameter-shift gradient estimator $\widehat G(\omega)$, the projected VQEG update is defined as follows:
\begin{align}
\label{eq:eg_1_2}
    \omega_{t+\frac12}
    &= \Pi_{\mathcal W}\!\left(\omega_t-\eta_t \widehat G(\omega_t)\right), \\
\label{eq:eg_1}
    \omega_{t+1}
    &= \Pi_{\mathcal W}\!\left(\omega_t-\eta_t \widehat G(\omega_{t+\frac12})\right).
\end{align}
Thus, each iteration uses two operator evaluations: one at the current point to form the predictor, and one at the extrapolated point to construct the corrected update. The detailed procedure of the proposed VQEG is in Algorithm~\ref{alg:qeg}.

\begin{algorithm}[t]
\caption{Variational Quantum Extragradient Algorithm}
\label{alg:qeg}
\begin{algorithmic}[1]
\REQUIRE $A\in\mathbb{R}^{m\times n}$; steps $T$; stepsizes $\{\eta_t\}$; $C>\|A\|_\infty$; PQCs $U_\theta$ ($q_r$ qubits) and $V_\phi$ ($q_c$ qubits); projection set $\mathcal{W}\subset\mathbb{R}^{d_r+d_c}$; shots per shifted evaluation $S$.
\STATE Set $M=2^{\lceil\log_2 m\rceil}$, $N=2^{\lceil\log_2 n\rceil}$; construct $\widetilde A$ via~\eqref{eq:dom_embed}.
\STATE Initialize $\omega_0=(\theta_0,\phi_0)\in\mathcal{W}$.
\FOR{$t=0,1,\dots,T-1$}
\STATE Estimate $\widehat G(\omega_t)$ using~\eqref{eq:ps_theta}--\eqref{eq:ps_phi_final} with $S$ shots per shifted circuit.
\STATE Predictor: $\omega_{t+\frac12} \leftarrow \Pi_{\mathcal{W}}(\omega_t - \eta_t \widehat G(\omega_t))$.
\STATE Estimate $\widehat G(w_{t+\frac12})$ using parameter-shift.
\STATE Corrector: $\omega_{t+1} \leftarrow \Pi_{\mathcal{W}}(\omega_t - \eta_t \widehat G(w_{t+\frac12}))$.
\STATE Evaluate $(x_{\theta_{t+1}},y_{\phi_{t+1}})$ and Nash gap certificate.
\ENDFOR
\STATE Return $\omega_T$ and/or an average of iterates; report certificate values.
\end{algorithmic}
\end{algorithm}

\section{CONVERGENCE ANALYSIS} 
\label{sec:Convergence}

\subsection{Assumptions}

\begin{assumption}
\label{ass:domain}
    $\mathcal{W}$ is compact and convex with diameter $d_w := \sup_{u,v \in\mathcal{W}} \|u-v\|_2 < \infty$.
\end{assumption}

\begin{assumption}
\label{ass:smooth}
    $\mathcal{L}$ is differentiable and $\nabla \mathcal{L}$ is $L$-Lipschitz:
    $\|\nabla\mathcal{L}(\omega) - \nabla\mathcal{L}(\omega')\|_2\le L\|\omega - \omega'\|_2$ for all $\omega, \omega' \in \mathcal{W}$.
\end{assumption}

\begin{assumption}
\label{ass:oracle}
    The gradient estimator $\widehat G(\omega)$ satisfies $\mathbb{E}[\widehat G(\omega) \mid \omega] = G(\omega)$ and $\mathbb{E}[\|\widehat G(\omega) - G(\omega)\|_2^2 \mid \omega] \le \sigma^2$ for all $\omega \in \mathcal{W}$.
\end{assumption}


\subsection{Stationary Convergence}
For any step size $\eta>0$, define the projected residual as:
\begin{equation}
\label{eq:proj_residual}
    \mathcal{R}_\eta(\omega) = \frac{1}{\eta}\Big(\omega - \Pi_{\mathcal{W}} \big(\omega - \eta \, G(\omega) \big) \Big).
\end{equation}
Then $\mathcal{R}_\eta(\omega) = 0$ if and only if $\omega$ is a first-order stationary point of the projected parametric game. Accordingly, $\| \mathcal{R}_\eta(\omega) \|_2$ serves as our stationarity measure.


\begin{lemma}[Predictor bound]
\label{lem:predictor_error}
    Let $z_t = \Pi_{\mathcal{W}}\!\left(\omega_t-\eta G(\omega_t)\right)$ denote the exact projected predictor generated by the true saddle operator, and define the stochastic gradient noise by $\xi_t := \widehat G_t(\omega_t)-G(\omega_t)$. Then the stochastic predictor $\omega_{t+\frac12}$ deviates from its exact counterpart by:
    \begin{align}
    \label{eq:predictor-diff}
        & \|\omega_{t+\frac12}-z_t\|_2 \le \eta \left\|\xi_t\right\|_2.
    \end{align}
    In particular, if Assumption~\ref{ass:oracle} holds, then
    \begin{align}
    \label{eq:predictor-diff-expect}
        & \mathbb{E}\!\left[\|\omega_{t+\frac12}-z_t\|^2_2\right] \le \eta^2 \sigma^2.
    \end{align}
\end{lemma}
\begin{proof}
    Please refer to our technical report \cite{duong2026vqegtech}. 
\end{proof}

\begin{lemma}[Corrector bound]
\label{lem:corrector_error}
    Let the stochastic gradient error at the corrector step be $\xi_t' := \widehat G_t'(\omega_{t+\frac12}) - G(\omega_{t+\frac12})$.
    Then the corrector iterate satisfies:
    \begin{align}
    \label{eq:corrector-vs-z}
        & \left\|\omega_{t+1}-z_t\right\| \le
        \eta L \|\omega_{t+\frac12}-\omega_t\|_2 + \eta \left\|\xi_t'\right\|_2.
    \end{align}
    Consequently, under Assumption~\ref{ass:oracle}, it holds:
    \begin{align}
    \label{eq:corrector-vs-z-exp}
        \mathbb{E}\!\left[\left\|\omega_{t+1} \!-\! z_t\right\|^2\right]
        \!\le\!
        2\eta^2 L^2 \mathbb{E}\!\left[\|\omega_{t+\frac12}-\omega_t\|^2_2\right] + 2\eta^2 \sigma^2.
\end{align}
\end{lemma}
\begin{proof}
    Please refer to our technical report \cite{duong2026vqegtech}. 
\end{proof}

\begin{lemma}[Residual-update bound]
\label{lem:r-vs-step}
    For every $t$, the projected residual satisfies the following bound:
    \begin{align}
    \label{eq:r-step-bound}
        & \eta^2 \left\|R_\eta(\omega_t)\right\|^2_2 \le 2\left\|\omega_{t+1}-\omega_t\right\|^2_2 + 2\left\|\omega_{t+1}-z_t\right\|^2_2.
    \end{align}
\end{lemma}
\begin{proof}
    Please refer to our technical report \cite{duong2026vqegtech}. 
\end{proof}

\begin{theorem}[Approximate stationarity]
\label{thm:stationarity}
    Suppose Assumptions~\ref{ass:domain}--\ref{ass:oracle} hold. Let $\{\omega_t\}_{t=0}^{T-1}$ be generated by the projected stochastic extragradient method with constant stepsize $\eta \le c/L$, where $c>0$ is a sufficiently small absolute constant. Then there exist constants $C_1,C_2>0$, depending on smoothness constant $L$ and an upper bound on the variation of $\mathcal{L}$ over $\mathcal{W}$, such that for a uniformly random index $\tau$ from $\{0,\dots,T-1\}$, the expected squared projected residual satisfies the following bound:
    \begin{align}
    \label{eq:stationary}
        & \mathbb{E} \big[\|\mathcal{R}_\eta(\omega_\tau) \|_2^2 \big] \;\le\; \frac{C_1}{\eta^2 T} + C_2\,\sigma^2.
    \end{align}
    In particular, when the saddle-gradient estimator is obtained via parameter-shift with $S$ shots per shifted circuit evaluation, we have $\sigma^2 = O(1/S)$, and hence, it follows that:
    \begin{align}
    \label{eq:main-thm-shots}
        & \mathbb{E}\!\left[\left\|R_\eta(\omega_\tau)\right\|^2_2\right] \le \frac{C_1}{\eta^2 T} + O\!\left(\frac{1}{S}\right).
    \end{align}
\end{theorem}

\textit{Proof:} The full proof is provided in~\cite{duong2026vqegtech}.

\begin{remark}
\label{rmk:sta}
    Theorem~\ref{thm:stationarity} is a \emph{nonconvex--nonconcave} guarantee. It proves convergence to \emph{approximate projected stationarity} of the \emph{parametric} game. Stationarity in parameter space, however, does not generally imply proximity to a global saddle point. Thus, the induced strategies $(x_{\theta_t},y_{\phi_t})$ are not guaranteed to achieve a small Nash gap in the original finite game.   
\end{remark}

\section{NUMERICAL RESULTS}
\label{sec:Results}
We evaluate the proposed projected VQEG algorithm on finite two-player zero sum matrix games with sizes $m=n \in \{2, 4, 5, 6, 8, 16, 32\}$ with different classes of games as \emph{dominant row instances, matching pennies instances, random instances}. For example, Fig.~\ref{fig:game_instance} shows two representative $8\times 8$ instances used in our experiments. The dominant-row game shown in Fig.~\ref{fig:Dominant}  exhibits a clear equilibrium structure, where the last row yields consistently high payoffs across all column actions. Fig. \ref{fig:Random}  shows a random instance with no obvious structure, representing a more challenging setting.
\begin{figure}[h!]
\vspace{-0.3cm}
     \subfigure[Dominant-row instance]{
        \hspace*{-0.1in} 
        \includegraphics[width=0.235\textwidth]{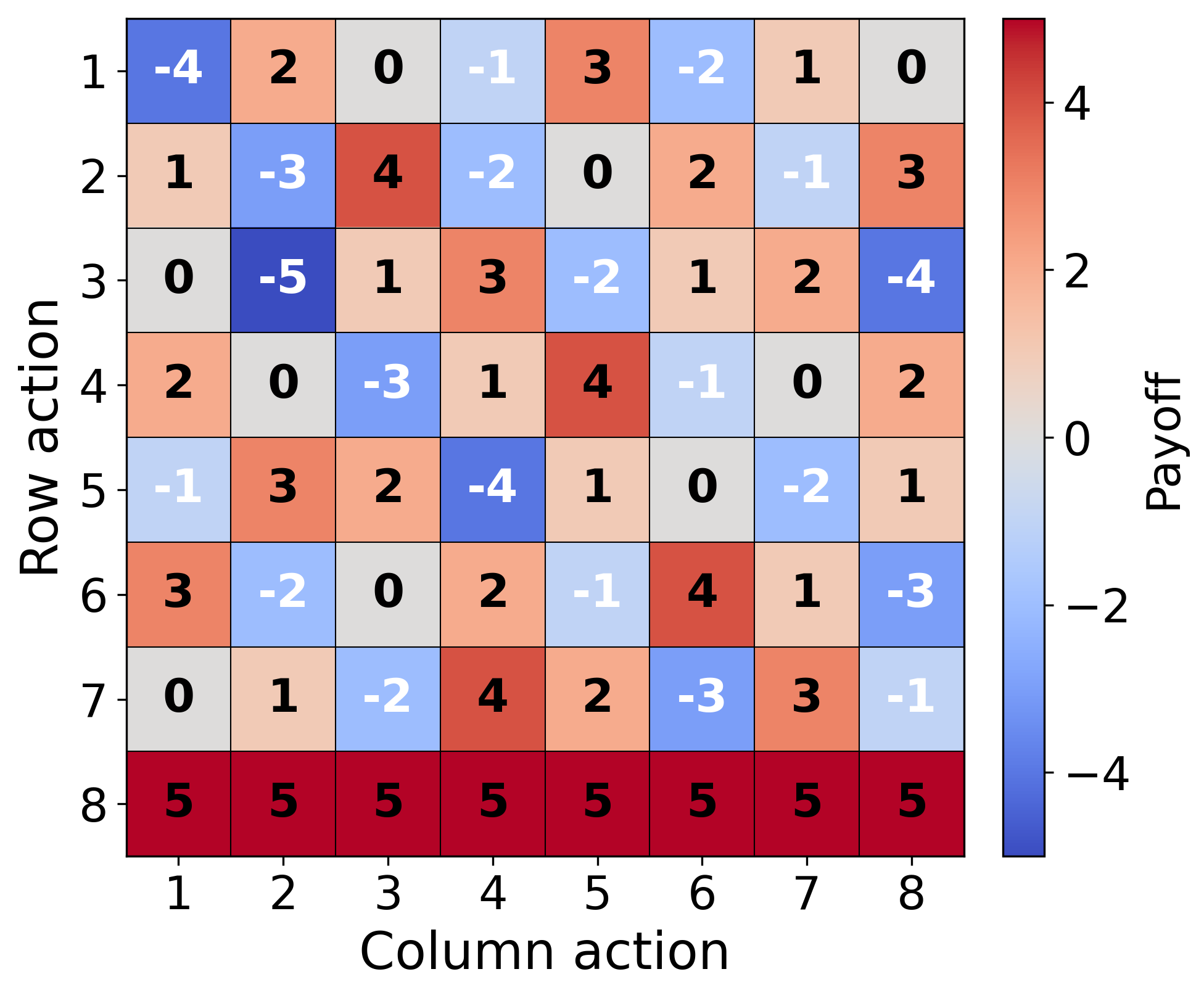}
	    \label{fig:Dominant} }
	\hspace*{-.1in}
    \subfigure[Random instance]{
	\includegraphics[width=0.235\textwidth]{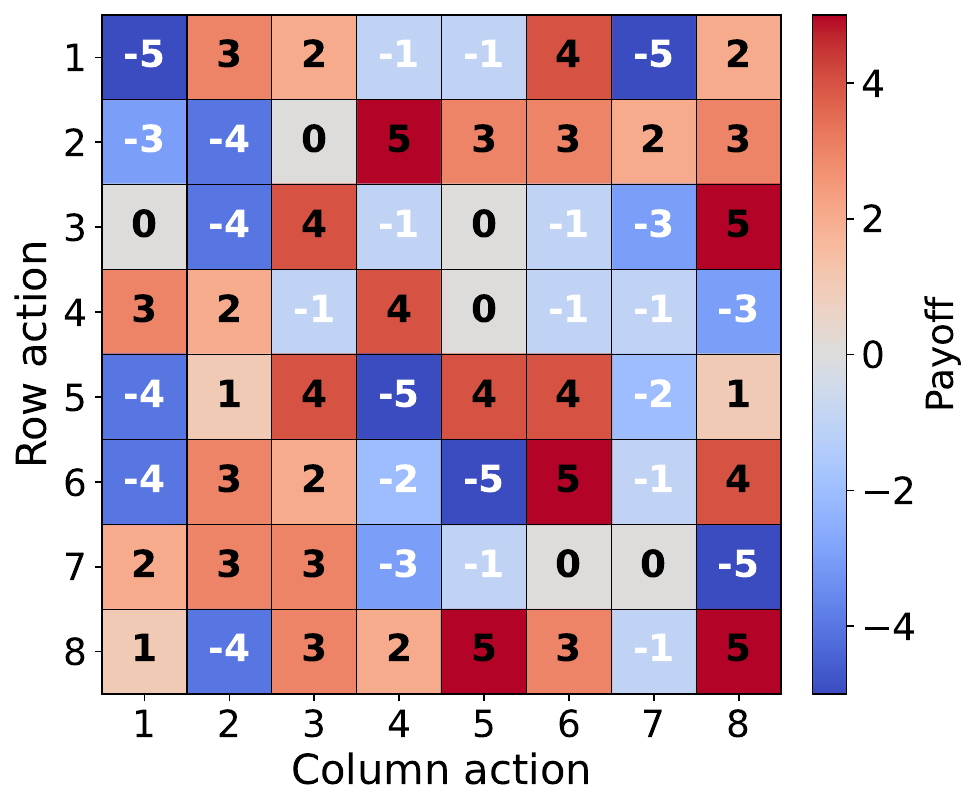}
	    \label{fig:Random} } 
        \vspace{-0.3cm}
	\caption{Illustration of payoff matrix for $8\times 8$ zero-sum game.}
	\label{fig:game_instance}
    \vspace{-0.3cm}
\end{figure}

For each instance, the exact NE and the corresponding game value $v^\star$ are computed using an LP solver and used as ground truth 
for evaluation. Performance is assessed using the Nash gap defined in~\eqref{eq:NE_gap_eq}, evaluated on the induced strategies $(x,y)$ in the original game. 
A solution is considered successful if $ \mathcal{G}(x, y) \le \varepsilon $, with  tolerance $\varepsilon = 5 \times 10^{-3}$.
Our circuit ansatz is constructed by which, within each layer, we apply single-qubit rotations $\mathrm{RY}$ and $\mathrm{RZ}$ on every qubit, followed by a nearest-neighbor $\mathrm{CZ}$ entangling ring. With $L$ layers, each player has $d_{\text{player}}=2qL$ trainable parameters. For example, for $m=n=4$ we have $q_r=q_c=2$, so each layer contributes $2q=4$ parameters per player.
\vspace{-0.3cm}
\begin{figure}[h!]
     \subfigure[Nash-gap convergence]{
        \hspace*{-0.1in} 
        \includegraphics[width=0.225\textwidth]{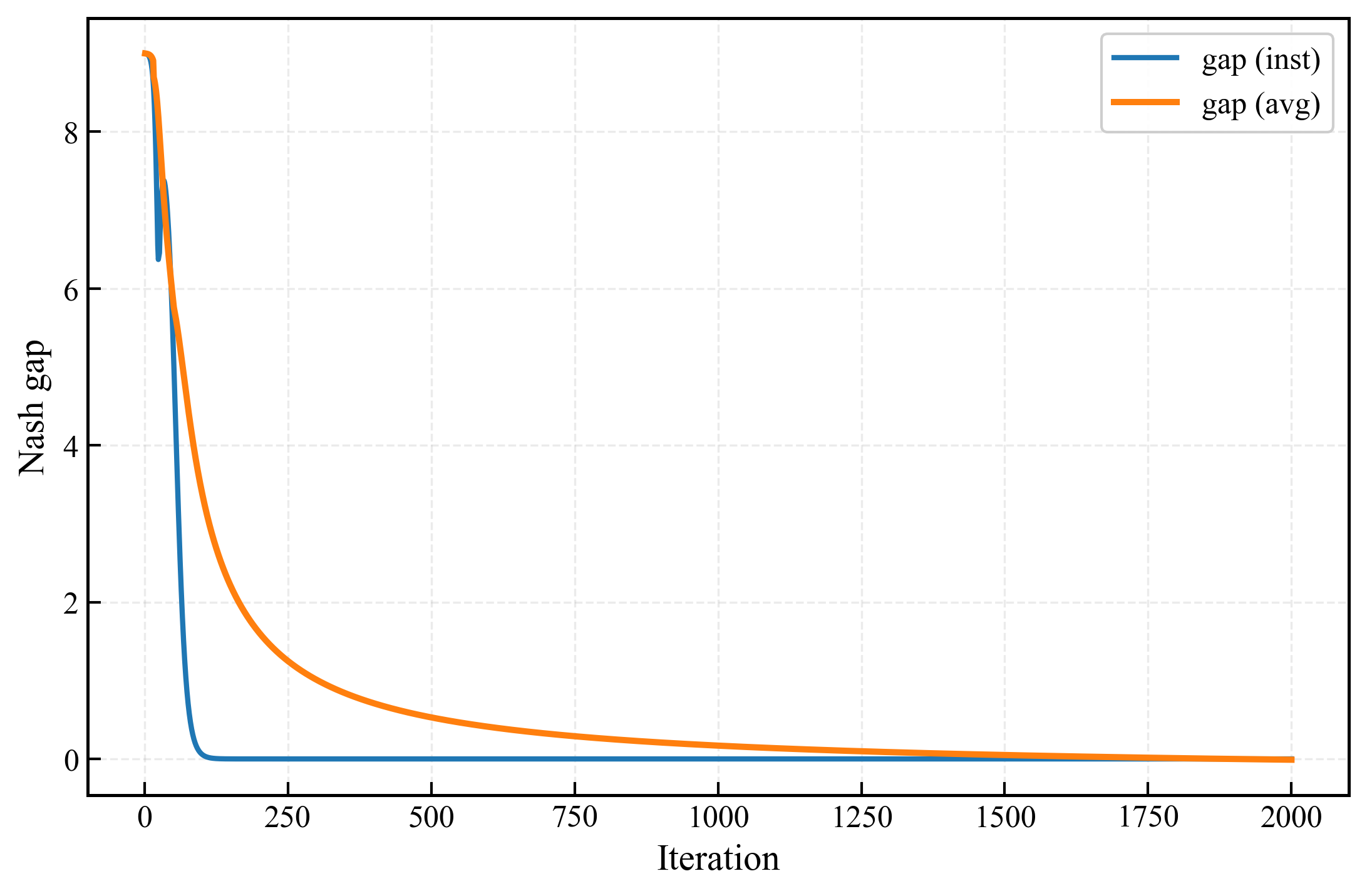}
	    \label{fig:NG} }
	\hspace*{-.1in}
    \subfigure[Leakage to dummy action]{
	\includegraphics[width=0.235\textwidth]{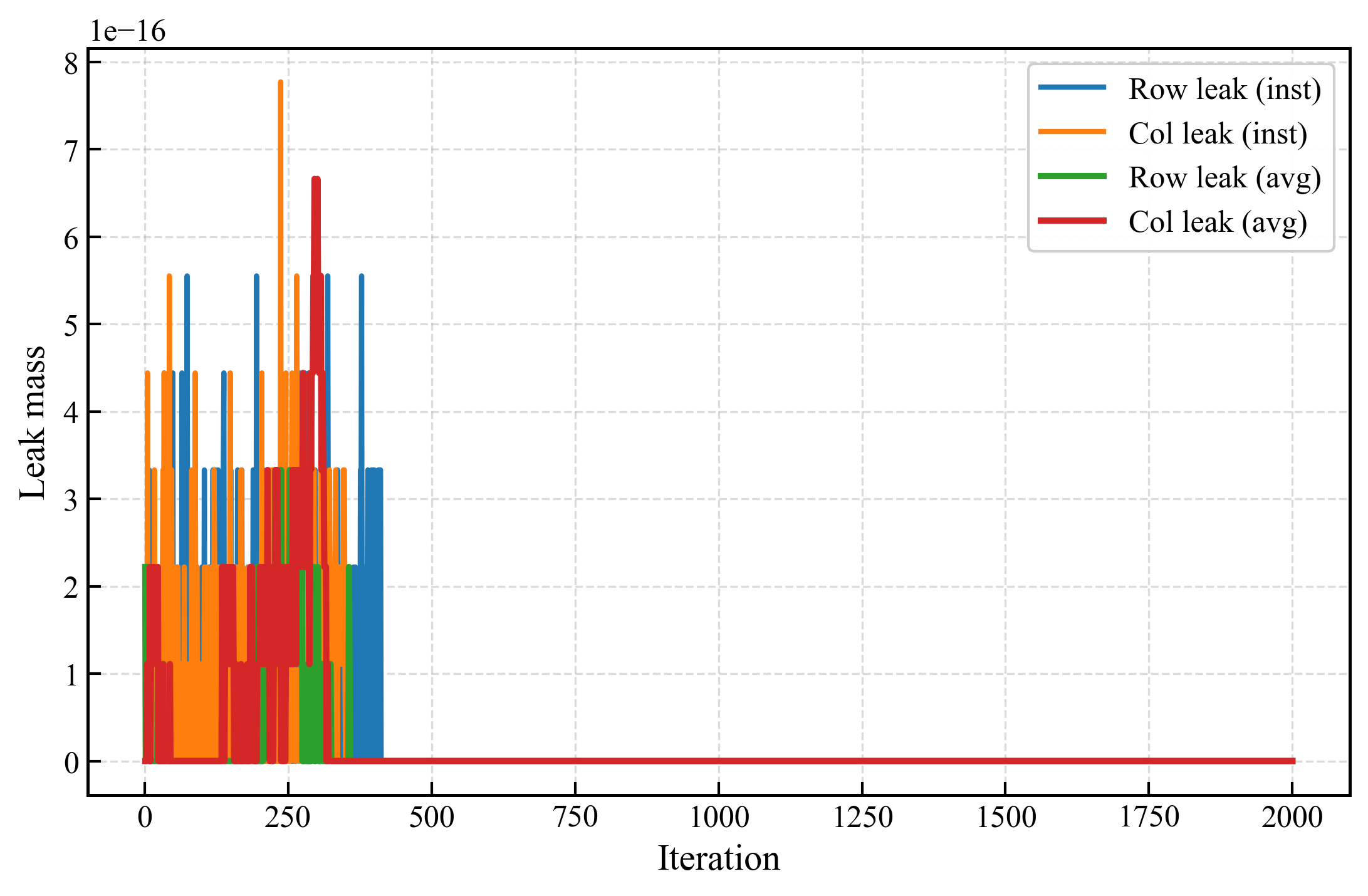}
	    \label{fig:Leak} } 
	\caption{VQEG convergence and leakage for $8 \times 8$ matrix game with dominant row startegy.}
	\label{fig:Dominant_4x4_Nash_gap}
\end{figure}

\vspace{-0.3cm}
\begin{table*}[t]
\centering
\caption{The VQEG Nash gap across different instances.}
\label{tab:avg_best}
\setlength{\tabcolsep}{3pt}
\renewcommand{\arraystretch}{1.15}
\begin{tabular}{llccccc}
\hline
Game & Size &
gap (avg) & PASS &
gap (best) & PASS &
VQEG-Entropy avg / best (PASS) \\
\hline
Dominant-row & \(4\times4\)   &
\(9.581\times10^{-6}\) & $\checkmark$ &
\(9.168\times10^{-6}\) & $\checkmark$ &
\(\approx 0\) / \(\approx 0\) (\checkmark/\checkmark) \\
Dominant-row & \(8\times8\)   &
\(1.615\times10^{-5}\) & \checkmark &
\(5.459\times10^{-6}\) & \checkmark &
\(7.616\times10^{-5}\) / \(3.470\times10^{-5}\) (\checkmark/\checkmark) \\
Dominant-row & \(16\times16\) &
\(6.340\times10^{-4}\) & $\checkmark$ &
\(\approx 0\) & \checkmark &
\(1.786\times10^{-11}\) / \(\approx 0\) (\checkmark/\checkmark) \\
Dominant-row & \(32\times32\) &
\(3.410\times10^{-4}\) & \checkmark &
\(2.000\times10^{-5}\) & \checkmark &
\(7.956\times10^{-6}\) / \(2.986\times10^{-9}\) (\checkmark/\checkmark) \\
\hline
Matching-pennies & \(4\times4\)   &
\(1.868\times10^{-10}\) & \checkmark &
\(6.661\times10^{-16}\) & \checkmark &
\(4.210\times10^{-7}\) / \(5.406\times10^{-7}\) (\checkmark/\checkmark) \\
Matching-pennies & \(8\times8\)   &
\(5.378\times10^{-7}\) & \checkmark &
\(1.133\times10^{-9}\) & \checkmark &
\(2.062\times10^{-3}\) / \(1.514\times10^{-4}\) (\checkmark/\checkmark) \\
Matching-pennies & \(16\times16\) &
\(1.953\times10^{-8}\) & \checkmark &
\(2.927\times10^{-12}\) & \checkmark &
\(7.639\times10^{-3}\) / \(2.188\times10^{-3}\) ($\times$/\checkmark) \\
Matching-pennies & \(32\times32\) &
\(3.125\times10^{-9}\) & \checkmark &
\(1.403\times10^{-13}\) & \checkmark &
\(1.184\times10^{-2}\) / \(7.749\times10^{-3}\) ($\times$/$\times$) \\
\hline
Randomize & \(4\times4\)   &
\(2.444\times10^{-5}\) & $\checkmark$ &
\(1.312\times10^{-5}\) & $\checkmark$ &
\(2.456\times10^{-1}\) / \(1.167\times10^{-1}\) ($\times$/$\times$) \\
Randomize & \(8\times8\)   &
\(3.515\times10^{-5}\) & $\checkmark$ &
\(1.910\times10^{-5}\) & $\checkmark$ &
\(1.989\times10^{-2}\) / \(1.096\times10^{-2}\) ($\times$/$\times$) \\
Randomize & \(16\times16\) &
\(2.875\times10^{-3}\) & $\checkmark$ &
\(3.174\times10^{-4}\) & $\checkmark$ &
\(2.770\times10^{-1}\) / \(2.286\times10^{-1}\) ($\times$/$\times$) \\
Randomize & \(32\times32\) &
\(2.154\times10^{-3}\) & $\checkmark$ &
\(2.027\times10^{-3}\) & $\checkmark$ &
\(3.912\times10^{-1}\) / \(3.891\times10^{-1}\) ($\times$/$\times$) \\
\hline
\end{tabular}
\end{table*}

Fig.~\ref{fig:Dominant_4x4_Nash_gap} illustrates the convergence behavior of VQEG on an $8 \times 8$ dominant-row instance. It is observed that the instantaneous Nash gap $\mathcal{G}_t$ decreases rapidly in the early iterations and then remains close to zero for thereafter, as displayed in Fig.~\ref{fig:NG}, indicating stable last-iterate convergence. 
The running-average gap decreases more gradually because it incorporates the suboptimal early iterates, but the final tail-averaged gap is effectively reaching equilibrium accuracy. 
Fig.~\ref{fig:Leak} further shows the leakage to dummy actions introduced by the padding construction. Both the row and column leakage remain about $10^{-16}$ throughout the run, with only minor fluctuations during the initial iterations. These results confirm that the proposed VQEG attains an equilibrium solution.

Table \ref{tab:avg_best} summarizes the performance of VQEG across different game classes and sizes. For dominant-row games, VQEG consistently achieves high-precision equilibria across all sizes up to $32 \times 32$, with Nash gaps often approaching numerical precision. Both average and best-case performance satisfy the $\epsilon$-equilibrium criterion. For matching-pennies games, VQEG remains stable and accurate, achieving very small Nash gaps across all sizes. 
While entropy-regularized variants degrade at larger scales, VQEG maintains strong performance, particularly in best-case results. Finally, for random games, although VQEG still satisfies the $\epsilon$-equilibrium criterion in many cases, the achieved Nash gaps are larger and performance degrades with increasing problem size. This highlights the impact of representation limitations and stochastic estimation noise in unstructured settings.

\section{CONCLUSIONS}
\label{sec:Conclusion}
We proposed a VQEG framework for computing equilibria in two-player zero-sum matrix games. The approach parameterizes mixed strategies via PQC-induced Born distributions and solves the resulting parametric saddle-point problem using projected extragradient updates with parameter-shift gradient estimation and explicit shot complexity considerations. 
To enable qubit-compatible representations without altering equilibrium structure, we introduced a dominated embedding that maps arbitrary $(m,n)$ games to power-of-two dimensions.
We established convergence to approximate projected stationarity under standard smoothness and bounded-variance assumptions, and evaluated performance using the game-space Nash gap as an equilibrium certificate.
Numerical results demonstrate that VQEG achieves high-precision equilibria on structured instances and scales to moderately large problems, while highlighting challenges in unstructured settings.
Future work will focus on improving circuit expressivity, variance reduction techniques, and algorithmic regularization to enhance scalability and robustness in more challenging game settings.

\bibliographystyle{IEEEtran}
\bibliography{Refs.bib}

\end{document}